\documentclass{article}

\usepackage{amsmath}
\usepackage{amssymb}
\usepackage{amsfonts}
\usepackage{graphicx}

\newcommand{\ket}[1]{\vert {#1} \rangle}  
\newcommand{\bra}[1]{\langle {#1} \vert}

\newcommand{\CNOT}{{\rm CNOT}}
\newcommand{\QCNOT}{\mathbb{CNOT}}

\title{Fuzzy approach for $\CNOT$ gate in quantum computation with mixed states}
\author{{\sc G. Sergioli}$^{1}$ and \ {\sc H. Freytes}$^{1,2}$}

\date{{\small
1. University of Cagliari, Viale Merello 92, 09123, Cagliari-Italy\\
2. Department of Mathematics UNR-CONICET, Av. Pellegrini 250, CP 2000, Rosario, Argentina.}}

\begin{document}

\maketitle

\begin{abstract}
In the framework of quantum computation with mixed states, a fuzzy representation of $\CNOT$ gate is introduced. In this representation, the incidence of
non-factorizability is specially investigated.
\end{abstract}

\begin{small}

  {\em PACS numbers: 03.67.Lx, 02.10.-v}

  {\em Keywords: $\CNOT$ quantum gate, quantum operations, non-factorizability.}

\end{small}


\newtheorem{theo}{Theorem}[section]
\newtheorem{definition}[theo]{Definition}
\newtheorem{lem}[theo]{Lemma}
\newtheorem{prop}[theo]{Proposition}
\newtheorem{coro}[theo]{Corollary}
\newtheorem{exam}[theo]{Example}
\newtheorem{rema}[theo]{Remark}{\hspace*{4mm}}
\newtheorem{example}[theo]{Example}
\newcommand{\proof}{\noindent {\em Proof:\/}{\hspace*{4mm}}}
\newcommand{\qed}{\hfill$\Box$}

\section*{Introduction}

The concept of quantum computing, introduced at the beginning of 1980s by Richard Feynman, is animated by the fact that quantum systems make possible new interesting forms of computational and communication processes. In fact, quantum computation can be seen as an extension of classical computation where new primitive information resources are introduced. Especially, the concept of quantum bit (qubit for short) which is the quantum counterpart of the classical bit. Thus, new forms of computational processes are developed in order to operate with these new information resources. 
In classical computation, information is encoded by a series of bits represented by the binary values $0$ and $1$. Bits are manipulated via ensemble of logical gates such as {\it NOT, OR, AND} etc, that form a circuit giving out the result of a calculation.

Standard quantum computing is based on quantum systems described by finite dimensional Hilbert spaces, such as ${\mathbb{C}}^2$, that is the
two-dimensional Hilbert space where a generic {\it qubit} lives. Hence, a qubit is represented by a unit vector
in ${\mathbb{C}}^2$ and, generalizing to a positive integer $n$, {\it $n$-qubits} are represented by unit vectors in
${\mathbb{C}}^{2^n}={\otimes^n} {\mathbb{C}}^2 = {\mathbb{C}}^2 \otimes
{\mathbb{C}}^2 \otimes \ldots \otimes {\mathbb{C}}^2$ ($n$- times). Similarly to the classical case, we can introduce and study the behavior of a number of {\it quantum
logical gates} (hereafter quantum gates for short) operating on qubits. As in the classical case, a quantum circuit is identified with an appropriate composition of quantum gates. They
are mathematically represented by unitary operators acting on pure states. In this framework only reversible processes are considered. But for many reasons this restriction is unduly. On the one hand, it does not encompass realistic physical states described by mixtures. In fact, a quantum system rarely is in a pure state. This may be
caused, for example, by the incomplete efficiency in the preparation procedure and also by manipulations on the system as measurements over pure states, both of which produce statistical mixtures. It motivated the study of a more general model of quantum computational processes, where pure states and unitary operators are replaced by density operators and quantum operations, respectively. This more general approach, where not only reversible transformations are considered, is called {\it quantum computation with mixed states} \cite{AKN, DF, FD, FSA, GUD1}. In this powerful model, fuzzy logic can play an important role to describe certain aspects of the combinational structures of quantum circuits. 

Our work is motivated on a fuzzy behavior of the $\CNOT$ quantum gate that arises when the model of quantum computation with mixed states is considered.

The paper is organized as follows: Section 1 contains generalities about tensor product structures to describe bipartite quantum systems. In Section 2
we recall some basic notions about the model of quantum computation with mixed states. In Section 3 we study fuzzy aspects of the $\CNOT$ when factorized inputs are considered. Section 4 generalizes the precedent section by considering non-factorized states. Finally in Section 5 we establish necessary and sufficient condition on the input of $\CNOT$ for which the factorizability of density operators in $\otimes^2{\mathbb{C}}^2$ is preserved by $\CNOT$.

\section{Bipartite systems} \label{HOL}

The notion of {\it state of a physical system} is familiar from its use in classical mechanics, where it is linked to the initial
conditions (the initial values of position and momenta) which determine the solutions of the equations of motion for the system. For
any value of time, the state is represented by a point in the phase space. In classical physics, compound systems can be decomposed into their subsystems. Conversely, individual systems can be combined to give overall composite systems. In this way, a classical global system is completely described in terms of the states of its subsystems and their mutual dynamic interactions. This property about classical systems is known as {\it separability principle}. From a mathematical point of view, the separability condition of classical systems comes from the fact that states of compound systems are represented as direct sum of the states of their subsystems. 

In quantum mechanics the  description of the state  becomes substantially modified. A quantum state can be either pure or mixed. A {\it pure state} is described by a unit vector in a Hilbert space and it is denoted by $\vert \varphi \rangle$ in Dirac notation.  When a quantum system is not in a pure state, it is represented by a probability distibution of pure states, the so called \emph{mixed state}.
Mixed states are mathematically modelled by {\it density operators} on a Hilbert space, i.e. positive self-adjoint, trace class operators. In terms of density operators, a pure state $\vert \psi \rangle$ can be represented as the projector $\rho = \vert \psi \rangle \langle \psi \vert$. Thus, in quantum theory, the most general description of a quantum states is encoded by density operators.

In quantum mechanics a system consisting of many parts is represented by the tensor product of the Hilbert spaces associated with the individual parts. We restrict our investigation to compound systems living in the bipartite Hilbert space of the form ${\cal H}_1 \otimes {\cal H}_2$, where ${\cal H}_1$ and ${\cal H}_2$ are finite dimensional.  
But not all density operators on ${\cal H}_1 \otimes {\cal H}_2$ are expressible as $\rho = \rho_1 \otimes \rho_2$, where $\rho_i$ is a density operator living in ${\cal H}_i$, for $i\in\{1,2\}$. Thus, there exist properties of quantum systems that characterize the whole system but that are not reducible to the local properties of its parts. Unlike classical physics, compound quantum systems can violate the separability principle. 

From a mathematical point of view, the origin of this difference between classical and quantum systems arises from the tensor product structure related to the Hilbert spaces. More precisely, the non-factorizability property of quantum states is related to the fact that the direct sum of ${\cal H}_1$ and ${\cal H}_2$ is a proper subset of ${\cal H}_1 \otimes {\cal H}_2$. 

In what follows we provide a formal description of this instance of non-factorizability of quantum states in compound systems of the form ${\cal H}_1 \otimes {\cal H}_2$.

Due to the fact that the Pauli matrices 
 \begin{equation*}
  \sigma_1 =
  \left[\begin{array}{cc}
    0 & 1 \\
    1 & 0
  \end{array}\right],
  \quad
  \sigma_2 =
  \left[\begin{array}{cc}
    0 & -i \\
    i & 0
  \end{array}\right],
  \quad
  \sigma_3 =
  \left[\begin{array}{cc}
    1 & 0 \\
    0 & -1
  \end{array}\right]
\end{equation*}
and $I$ are a basis for the set of operators over ${\mathbb{C}}^2$, an arbitrary density operator $\rho$ over ${\mathbb{C}}^2$ may be represented
as $$\rho=\frac{1}{2}(I+s_1\sigma_1+s_2\sigma_2+s_3\sigma_3)$$ where $s_1,s_2$ and $s_3$ are three real numbers such $s_1^2+s_2^2+s_3^2\le 1$. The triple $(s_1,s_2,s_3)$ represents the point of the Bloch sphere that is uniquely associated to $\rho$. A similar canonical representation can be obtained for any $n$-dimensional Hilbert space by using the notion of generalized Pauli-matrices.

\begin{definition}
{\rm Let $\mathcal H$ be a $n$-dimensional Hilbert space and $\{\ket{\psi_1},\ldots,\ket{\psi_n}\}$ be the canonical othonormal basis of $\mathcal H$. Let $k$ and $j$ be two natural numbers such that: $1\le k < j \le n$. Then, the {\it generalized Pauli-matrices} are defined as follows: 
 $$^{(n)}\sigma_1^{[k,j]}=
\ket{\psi_j}\bra{\psi_k}+
\ket{\psi_k}\bra{\psi_j}$$
 $$^{(n)}\sigma_2^{[k,j]}=
i(\ket{\psi_j}\bra{\psi_k}-\ket{\psi_k}\bra{\psi_j})$$
and for $1\le k \le n-1$
$$^{(n)}\sigma_3^{[k]}= \sqrt{\frac{2}{k(k+1)}}(\ket{\psi_1}\bra{\psi_1}+\cdots+\ket{\psi_k}\bra{\psi_k}-k\ket{\psi_{k+1}}\bra{\psi_{k+1}}).$$
}
\end{definition}

If $\mathcal H=\mathbb C^2$ one immediately obtains:  $^{(2)}\sigma_1^{[1,2]}= \sigma_1$, $^{(2)}\sigma_2^{[1,2]}= \sigma_2$ and $^{(2)}\sigma_3^{[1]}= \sigma_3.$

Let $\rho$ be a density operator of the $n$-dimensional Hilbert space ${\cal H}$. For any $j$, where $1\leq j \leq n^2-1$, let $$s_j(\rho) = tr(\rho \sigma_j) .$$ The sequence $\langle s_1(\rho)\ldots s_{n^2-1}(\rho) \rangle$ is called the {\it generalized Bloch vector} associated to $\rho$, in view of the following well known result:

\begin{theo}\label{BLOCHVECT}{\rm \cite{SM}}
Let $\rho$ be a density operator of the $n$-dimensional Hilbert space ${\cal H}$ and let $\sigma_j\in \mathfrak P_n$. Then $\rho$ can be canonically represented as follows: $$\rho = \frac{1}{n}I^{(n)} + \frac{1}{2}\sum_{j=1}^{n^2-1}s_j(\rho)\sigma_j $$ where $I^{(n)}$ is the $n\times n$ identity matrix.
\qed

\end{theo}

A kind of converse of Theorem \ref{BLOCHVECT} reads: a matrix $\rho$ having the form $\rho = \frac{1}{n}I^{(n)} + \frac{1}{2}\sum_{j=1}^{n^2-1}s_j(\rho)\sigma_j $ is a density operator if an only if its eigenvalues are non-negative.

Let us consider the Hilbert space ${\cal H} = {\cal H}_a \otimes {\cal H}_b$. For any density operator $\rho$ on $\cal H$, we denote by 
$\rho_a$ the partial trace of $\rho$ with respect to the system ${\cal H}_b$ (i.e. $\rho_a = tr_{{\cal H}_b}(\rho)$) and by $\rho_b$ the partial trace of $\rho$ with respect to the system ${\cal H}_a$ (i.e. $\rho_b = tr_{{\cal H}_a}(\rho)$). For the next  developments it is useful to recall the following technical result:

\begin{lem}\label{TRACE}
Let $\rho $ be a density operator in a $n$-dimensional Hilbert space ${\cal H} = {\cal H}_a \otimes {\cal H}_b$ where $dim({\cal H}_a) = m$ and $dim({\cal H}_b) = k$.  If we divide $\rho$ in $m \times m$ blocks $B_{i,j}$, each of them is a $k$-square matrix, then:

\begin{eqnarray*}
\rho_a & = &tr_{{\cal H}_b}(\rho)  =\left[\begin{array}{cccc}
      tr B_{1,1} & tr B_{1,2} & \ldots & tr B_{1,m}  \\
      tr B_{2,1} & tr B_{2,2} & \ldots & tr B_{2,m}  \\
      \vdots & \vdots & \vdots & \vdots  \\
      tr B_{m,1} & tr B_{m,2} & \ldots & tr B_{m,m}   \\
          \end{array}\right] \\
\rho_b & = &tr_{{\cal H}_a}(\rho) =\sum_{i=1}^m B_{i,i}.
\end{eqnarray*}
\qed
\end{lem}

\begin{definition}
{\rm Let $\rho $ be a density operator in a Hilbert space ${\cal H}_m \otimes {\cal H}_k$ such that $dim({\cal H}_m) = m$ and $dim({\cal H}_k) = k$. Then $\rho$ is said to be {\it $(m,k)$-factorizable} iff $\rho = \rho_m \otimes \rho_k$ where $\rho_m$ is a density operator in ${\cal H}_m$ and $\rho_k$ is a density operator in ${\cal H}_k$.   
}
\end{definition}

It is well known that, if $\rho$ is $(m,k)$-factorizable as $\rho = \rho_m \otimes \rho_k$, this factorization is unique and $\rho_m$ and $\rho_k$ correspond to the reduced states of $\rho$ on ${\cal H}_m$ and ${\cal H}_k$, respectively.\\

Suppose that ${\cal H} = {\cal H}_a\otimes {\cal H}_b$, where $dim({\cal H}_a) = m$ and $dim({\cal H}_b) = k$. Let us consider  the generalized Pauli matrices $\sigma_1^a,\ldots, \sigma_{m^2-1}^a $ and $\sigma_1^b,\ldots, \sigma_{k^2-1}^b $ arising from  ${\cal H}_a$ and ${\cal H}_b$, respectively. 

If we define the following coefficients: $$M_{j,l}(\rho) = tr(\rho [\sigma_j^a \otimes \sigma_l^b]) - tr(\rho [\sigma_j^a \otimes I^{(k)}])tr(\rho [I^{(m)}\otimes\sigma_l^b])$$ and if we consider the matrix ${\bf M}(\rho)$ defined as $${\bf M}(\rho)=\frac{1}{4} \sum_{j=1}^{m^2-1}\sum_{l=1}^{k^2-1} M_{j,l}(\rho)(\sigma_j^a \otimes \sigma_l^b)$$ then  ${\bf M}(\rho)$ represents the ``additional component" of $\rho$ when $\rho$ is not a factorized state. More precisely, we can establish the following proposition:

\begin{prop}\label{decomposition}{\rm\cite{SM}}
Let $\rho$ be a density operator in ${\cal H} = {\cal H}_a\otimes {\cal H}_b$. 
$$\rho = \rho_a \otimes \rho_b + {\bf M}(\rho).$$
\qed
\end{prop}

The above proposition gives a formal representation of the instance of holism mentioned at the beginning of the section. In fact, a state $\rho$ in  
${\cal H}_a\otimes {\cal H}_b$ does not only depend on its reduced states $\rho_a$ and $\rho_b$, but also the summand ${\bf M}(\rho)$ is involved. Let us notice that ${\bf M}(\rho)$ is not a density operator and then it does not represent a physical state.

\section{Quantum computation with mixed states}\label{QLMS}

As already mentioned, a qubit is a pure state in the Hilbert space
${\mathbb{C}}^2$. The standard orthonormal basis $\{ \vert 0 \rangle ,
\vert 1 \rangle \}$ of ${\mathbb{C}}^2$, where $\vert 0 \rangle = (1,0)^\dagger$
and $\vert 1 \rangle = (0,1)^\dagger$, is generally called
\textit{logical basis}. This name refers to the fact that the logical truth is
related to $\vert 1 \rangle$ and the falsity to $\vert 0 \rangle$. Thus,
pure states $\vert \psi \rangle$ in ${\mathbb{C}}^2$ are superpositions of the basis vectors $\vert
\psi \rangle = c_0\vert 0 \rangle + c_1 \vert 1 \rangle$, where $c_0$ and $c_1$ are complex numbers such that $\vert
c_0 \vert^2 + \vert c_1 \vert^2 = 1$. 
Recalling the Born rule, any qubit $\vert \psi \rangle = c_0\vert 0
\rangle + c_1 \vert 1 \rangle$ may be regarded as a piece of
information, where the numbers $\vert c_0 \vert ^2$ and $\vert c_1 \vert ^2$ correspond to
the probability-values associated to the information described by the basic
states $\vert 0 \rangle$ and $\vert 1 \rangle$, respectively.  Hence, we confine our interesting to the probability value  $p(\vert \psi
\rangle) = \vert c_1 \vert ^2$, that is related to the basis vector
associated with the logical truth.

Arbitrary quantum computational states live in ${\otimes^n} {\mathbb{C}}^2$. A special basis,   called the
$2^n$-{\it computational basis}, is chosen for ${\otimes^n}
{\mathbb{C}}^2$. More precisely, it consists of the $2^n$ orthogonal
states $\vert \iota \rangle$, $0 \leq \iota \leq 2^n$ where $\iota$
is in binary representation and $\vert \iota \rangle$ can be seen as
tensor product of states (Kronecker product) $\vert \iota \rangle =
\vert \iota_1 \rangle \otimes \vert \iota_2 \rangle \otimes \ldots
\otimes \vert \iota_n \rangle$, with $\iota_j \in \{0,1\}$. Then, a pure
state $\vert \psi \rangle \in {\otimes^n}{\mathbb{C}}^2$ is  a
superposition of the basis vectors $\vert \psi \rangle = \sum_{\iota
= 1}^{2^n} c_{\iota}\vert \iota \rangle$, with $\sum_{\iota =
1}^{2^n} \vert c_{\iota} \vert^2 = 1$.

As already mentioned, in the usual representation of quantum
computational processes, a quantum circuit is identified with an
appropriate composition of {\it quantum gates}, mathematically
represented by \textit{unitary operators} acting on pure states of a
convenient ($n$-fold tensor product) Hilbert space ${\otimes^n}
{\mathbb{C}}^2$ \cite{NIC}. 

In what follows we give a short description of the model of quantum computers with mixed states.

We associate to each vector of
the logical basis of ${\mathbb{C}}^2$ two density operators $P_0 = 
\vert 0 \rangle \langle 0 \vert$ and $P_1 = \vert 1 \rangle \langle 1
\vert$ that represent, in this framework, the falsity-property and the truth-property,
respectively.  Let us consider the operator $P_1^{(n)} =
\otimes^{n-1} I \otimes P_1$ on ${\otimes^n}{\mathbb{C}}^2$. 
By
applying the Born rule, we consider the probability of a density
operator $\rho$ as follows:
\begin{equation} \label{PROBDEF}
  p(\rho) = Tr(P_1^{(n)} \rho).
\end{equation}
Note that, in the particular case in which $\rho = \vert \psi \rangle
\langle \psi \vert$, where $\vert \psi \rangle = c_0\vert 0 \rangle +
c_1 \vert 1 \rangle$, we obtain that $p(\rho)= \vert c_1 \vert
^2$. Thus, this probability value associated to $\rho$ is the
generalization of the probability value considered for
qubits.
In the model of quantum computation with mixed states, the role of quantum gates is replaced by quantum operations. A \textit{quantum operation} \cite{K} is a linear operator ${\cal
  E}:{\cal L}(H_1)\rightarrow {\cal L}(H_2)$ - where ${\cal L}(H_i)$ is
the space of linear operators in the complex Hilbert space $H_i$ ($i=
1, 2$) - representable (following the first Kraus representation theorem) as ${\cal E}(\rho)=\sum_{i}A_{i}\rho
A_{i}^{\dagger }$, where $A_i$ are operators satisfying
$\sum_{i}A_{i}^{\dagger }A_{i}=I$. It can be
seen that a quantum operation maps density operators into density
operators. Each unitary operator $U$ has a natural correspondent quantum operation ${\cal O}_ U$ such that, for each density operator $\rho$, ${\cal O}_U(\rho) =
U\rho U^{\dagger}$. In this way, quantum operations are generalizations of unitary operators.
It provides a powerfull mathematical model where also irreversible  porcesses can be considered.

\section{$\CNOT$ quantum operation as fuzzy connective}\label{SECFUZZYCON}
As in classical case, also in quantum computation it is useful to implement some kind of \textquotedblleft if-then-else\textquotedblright operations. More precisely, it means that we have to consider the evolution of a set of qubits depending upon the values of some other set of
qubits. The gates that implement these kind of operations are called \textquotedblleft controlled gates\textquotedblright. The controlled gates we are interested on, is the {\it controlled-NOT gate} ($\CNOT$, for short). An usual application of the $\CNOT$ gate is to generate entangle states, starting from factorizable ones. This is a crucial step for quantum teleportation protocol and quantum cryptography.

The $\CNOT$ gate, takes two qbits as input, a control qbit and a target qbit, and performs the
following operation: 
\begin{itemize}
\item if the control qbit is $\vert 0 \rangle$, then $\CNOT$ behaves as the identity 
\item if the control bit is $\vert 1 \rangle$, then the target bit is 
flipped. 
\end{itemize}

Thus, $\CNOT$ is given by the unitary transformation $$\ket i\ket j\mapsto\ket i\ket {i\widehat{+}j}$$ where $i, j\in \{0,1\}$ and  $\widehat{+}$ is the sum modulo $2$. Note that, confining in the computational basis only, the behaviour of $\CNOT$ replaces the classical XOR connective. The matrix representation of $\CNOT$ is given by:  
\begin{equation} \label{matrixcnot}
\CNOT = \left(\begin{array}{cccc}
      1 & 0 & 0 & 0  \\
      0& 1 & 0 & 0  \\
      0 & 0 & 0 & 1  \\
      0& 0 & 1 & 0  \\
          \end{array}\right).
\end{equation}

Since $\CNOT$ is a unitary matrix, it naturally admits an extension as quantum operation. Noting that $\CNOT^\dagger = \CNOT$, its extension as quantum operation is given by: 
  
\begin{equation} \label{matrixqcnot}
\QCNOT(\rho\otimes\sigma) = \CNOT \hspace{0.1cm} (\rho\otimes\sigma) \hspace{0.1cm} \CNOT.  
\end{equation}

\begin{theo}\label{NotFact}
Let $\rho$, $\sigma$ be two density operators in ${\mathbb{C}^2}$. Then: $$p(\QCNOT(\rho\otimes\sigma)) =  (1-p(\rho))p(\sigma)+ (1-p(\sigma)) p(\rho). $$
\end{theo}

\begin{proof}
Let 
\begin{equation*}
  \rho =
  \left(\begin{array}{cc}
    1-a & r \\
    r^* & a
  \end{array}\right) \hspace{0.2cm} and
  \quad
  \sigma =
  \left(\begin{array}{cc}
    1-b & t \\
    t^* & b
  \end{array}\right)
\end{equation*}
be density operators in $\mathbb C^2$. It is easy to check that the diagonal elements of $\rho\otimes\sigma$ are $d_{11}=(1-a)(1-b)$, $d_{22}=(1-a)b$, $d_{33}=a(1-b)$ and $d_{44}=ab$.
Similarly, the diagonal elements of $\QCNOT(\rho\otimes\sigma)$ are: $d'_{11}=d_{11}$, $d'_{22}=d_{22}$ and $d'_{33}=d'_{44}$ and $d'_{44}=d_{33}$.  Thus
\begin{eqnarray*}
p(\QCNOT(\rho\otimes\sigma)) &=& d'_{22} + d'_{44} \\
&=& (1-a)b+b(1-a) \\
&=& (1-p(\rho))p(\sigma)+ (1-p(\sigma)) p(\rho).
\end{eqnarray*}

\qed
\end{proof}

The above theorem allows us to consider $\QCNOT$ as a fuzzy connective in accord to the probability value $p(\QCNOT(-\otimes -))$. 

In fact: let $x,y \in [0,1]$; the usual product operation $x \cdot y$ in the unitary real interval defines the conjunction in the fuzzy logical system called {\it Product Logic} \cite{CT}. The operations $\neg_{\L} x = 1-x$ and $x\oplus y = \min\{x+y,1\}$ define the negation and the disjunction of the infinite value \L ukasiewicz calculus respectively \cite{CDM}. The operations  $\langle \cdot, \oplus, \neg_{\L} \rangle $ endow the interval $[0,1]$ of an algebraic structure known as {\it Product $MV$-algebra} ($PMV$-algebra, for short) \cite{DD,MONT2}. In this case the $PMV$-algebra $\langle [0,1] \cdot, \oplus, \neg_{\L} \rangle $ is the standard model of the a fuzzy logic system, called {\it Product Many Valued Logic}. 

If $\rho$, $\sigma$ are two density operators in ${\mathbb{C}^2}$, then $(1-p(\rho))p(\sigma)+p(\rho)(1-p(\sigma)) \leq 1$. Thus, $p(\QCNOT(\rho\otimes\sigma))$ can be expressed in terms of $PMV$-operations. More precisely:   
\begin{eqnarray*}
p(\QCNOT(\rho\otimes\sigma)) &=& (1-p(\rho))p(\sigma)+ (1-p(\sigma)) p(\rho) \\
&=& (\neg_{\L}p(\rho) \cdot p(\sigma)) \oplus (\neg_{\L}p(\sigma) \cdot p(\rho)). 
\end{eqnarray*}

In this way $\QCNOT$ can be relate to the fuzzy connective given by the $PMV$-polynomial term $(\neg_{\L} x \cdot y) \oplus (\neg_{\L} y \cdot x)$, establishing a link between $\QCNOT$ and a fuzzy logic system. Let us notice that there are other quantum gates admitting a similar fuzzy representation \cite{FS, FSA}.\\

In Figure \ref{pcnot} we show the behavior of $p(\QCNOT(-\otimes-)$ as a fuzzy connective.

\begin{figure}\caption{$p(\QCNOT(\rho\otimes\sigma))$}\label{pcnot}
\centering
\includegraphics[scale=0.4]{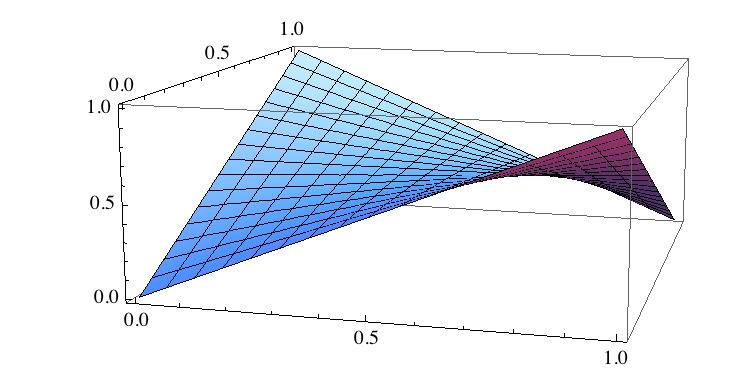}

\end{figure}

\section{$\QCNOT$ on general density operators}

In the precedent section we have introduced the behaviour of the $\QCNOT$ gate on factorized states of the form $\rho\otimes\sigma$. For a more general approach, we now assume that the input state can be any arbitrary mixed state $\rho$ in $\otimes^2{\mathbb C}^2$.   
We remark how this kind of studies suggests a holistic form of quantum logics \cite{BDGLS, SDGLL}. Let $\rho$ be a density operator in $\otimes^2{\mathbb C}^2$ and $\rho_1, \rho_2$ the reduced states of $\rho$. Since $\QCNOT$ is linear, by Proposition \ref{decomposition}, we have that 

\begin{equation}\label{SUMDEC} 
\QCNOT(\rho) = \QCNOT(\rho_1 \otimes \rho_2 ) + \CNOT({\bf M}(\rho))\CNOT.
\end{equation}

The summand $\QCNOT(\rho_1 \otimes \rho_2 )$ will be called the {\it fuzzy component} of $\QCNOT(\rho)$ and we denote by ${\cal C}(\rho)$ the quantity  $\CNOT({\bf M}(\rho))\CNOT$.

\begin{theo}\label{ENTCNOT}
Let $\rho$ be a density operator in $\mathbb C^4$ such that 
$$\rho=(r_{ij})_{1\leq i,j \leq 2^2} =\left(\begin{array}{cccc}
      r_{11} & r_{12} & r_{13} & r_{14}  \\
      r_{21}& r_{22} & r_{23} & r_{24}  \\
      r_{31} & r_{32} & r_{33} & r_{34}  \\
      r_{41}& r_{42} & r_{43} & r_{44}  \\
          \end{array}\right). 
$$
Then

\begin{enumerate}
\item
$p(\QCNOT(\rho)) = r_{22}+r_{33}$,

\item
$p(\QCNOT(\rho_1\otimes\rho_2)) = (r_{11} + r_{22})(r_{22} + r_{44}) + (r_{11} + r_{33})(r_{33} + r_{44})$,

\item
$-1 \leq Tr(P_1 {\cal C}(\rho)) = 2(r_{22}r_{33} - r_{11}r_{44}) \leq 1$, 

\item
$Tr(P_1 {\cal C}(\rho)) = \frac{1}{2}$ iff, $r_{11} = r_{44} = 0$ and $r_{22} = r_{33} = \frac{1}{2}$, 

\item
$Tr(P_1 {\cal C}(\rho)) = - \frac{1}{2}$ iff, $r_{11} = r_{44} = \frac{1}{2}$ and $r_{22} = r_{33} = 0$.

\end{enumerate}

\end{theo}

\begin{proof} 1) It is immediate to see that $diag(\QCNOT(\rho)) = (r_{11}, r_{22}, r_{44}, r_{33}) $. Then $p(\QCNOT(\rho)) = tr(P_1 \QCNOT(\rho)) = r_{22}+r_{33}$.

2) By Lemma \ref{TRACE} we have that
$$
\rho_1 =
  \left(\begin{array}{cc}
    r_{11} + r_{22} & r_{13} + r_{24} \\
    r_{31} + r_{42} & r_{33} + r_{44}
  \end{array}\right) \hspace{0.1cm} and \hspace{0.1cm}
  \rho_2 =
  \left(\begin{array}{cc}
    r_{11} + r_{33} & r_{12} + r_{34} \\
    r_{21} + r_{43} & r_{22} + r_{44}
  \end{array}\right).
$$

By Theorem \ref{NotFact} we have that 
\begin{eqnarray*}
p(\QCNOT(\rho_1\otimes \rho_2)) &=& (1-p(\rho_1))p(\rho_2)+ (1-p(\rho_2)) p(\rho_1) \\
&=& (1- (r_{33} + r_{44})(r_{22} + r_{44}) + (1- r_{22} + r_{44})(r_{33} + r_{44})\\ 
&=& (r_{11} + r_{22})(r_{22} + r_{44}) + (r_{11} + r_{33})(r_{33} + r_{44}).
\end{eqnarray*}

3,4,5) By Proposition \ref{decomposition}, 
\begin{eqnarray*}
Tr(P_1 {\cal C}(\rho)) &=&  p(\QCNOT(\rho)) -  p(\QCNOT(\rho_1\otimes\rho_2)) \\
&=& r_{22}+r_{33} - (r_{11} + r_{22})(r_{22} + r_{44}) - (r_{11} + r_{33})(r_{33} + r_{44})\\ 
&=& 2(r_{22}r_{33} - r_{11}r_{44}).
\end{eqnarray*}

Note that $Tr(P_1 {\cal C}(\rho))$ assumes the maximum value when $r_{11}r_{44} = 0$. If $r_{11} = 0$ then $1 = r_{2,2} + r_{3,3} + r_{4,4}$ and the maximum of $r_{2,2} r_{3,3}$ occurs when $r_{44} = 0$. It implies that $r_{2,2} + r_{3,3} = 1$. Thus $\max\{r_{2,2} r_{3,3}\}$ occurs when $r_{2,2} = r_{3,3} = \frac{1}{2}$. In this way, $\max\{Tr(P_1 {\cal C}(\rho))\}$ occurs when $r_{11} = r_{44} = 0$ and $r_{22} = r_{33} = \frac{1}{2}$ and $\max\{Tr(P_1 {\cal C}(\rho))\} = \frac{1}{2}$. With a similar argument we prove that $\min\{Tr(P_1 {\cal C}(\rho))\} = -\frac{1}{2}$ and it occurs when $r_{11} = r_{44} = \frac{1}{2}$ and $r_{22} = r_{33} = 0$.

\qed
\end{proof}

\begin{example}
{\rm Werner states provide an interesting example to show the behavior of $\QCNOT$ on a non-factorized states. Werner states, originally introduced in \cite{WER} for two particles to distinguish between classical correlation and the Bell inequality satisfaction, have many interests for their applications in quantum information theory. Examples of this, are entanglement teleportation via Werner states \cite{LEEKIM}, the study of deterministic purification \cite{SHORT}, etc. Werner states  in $\otimes^2\mathbb C^2$ are generally represented by the following expression:
$$\rho_w(\alpha) = \frac{1}{4} \left(\begin{array}{cccc}
      1- \alpha & 0 & 0 & 0  \\
      0 & 1 + \alpha & -2 \alpha & 0  \\
      0 & -2 \alpha & 1+ \alpha & 0  \\
      0& 0 & 0 & 1- \alpha  \\
          \end{array}\right). 
$$
where $\alpha \in [0,1]$. By Theorem \ref{ENTCNOT} we have that:

\begin{enumerate}
\item[]
$p(\QCNOT(\rho_w(\alpha))) = \frac{1+\alpha}{2}$,

\item[]
$p(\QCNOT({\rho_w(\alpha)}_1 \otimes {\rho_w(\alpha)}_2 )) = \frac{1}{2}$,

\item[]
$Tr(P_1 {\cal C}(\rho_w(\alpha)))) = \frac{\alpha}{2}$.
\end{enumerate}

Note that, for each $\alpha \in [0,1]$ the probability value of the fuzzy component $p(\QCNOT({\rho_w(\alpha)}_1 \otimes {\rho_w(\alpha)}_2 ))$ does not change. Thus, the variation of probability value $p(\QCNOT(\rho_w(\alpha)))$ is ruled by the variation of $Tr(P_1 {\cal C}(\rho_w(\alpha))))$.  The Figure \ref{Xor} shows the incidence of $Tr(P_1 {\cal C}(\rho_w(\alpha))))$ in the probability value $p(\QCNOT(\rho_w(\alpha)))$  when the parameter $w$ varies. 
}
\end{example}

\begin{figure}\caption{Incidence of $Tr(P_1 {\cal C}(\rho_w(\alpha))))$}\label{Xor}
\centering
\includegraphics[scale=0.4]{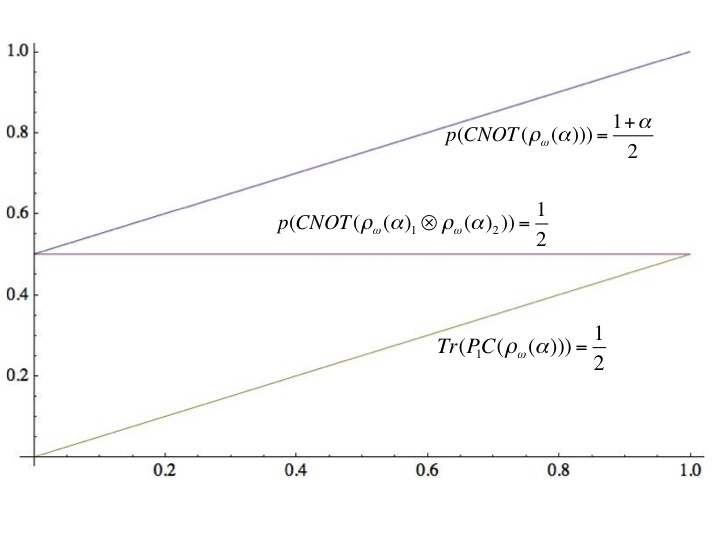}

\end{figure}







\section{Preservation of factorizability by $\QCNOT$}

As we noted at the beginning of Section \ref{SECFUZZYCON}, $\CNOT$ allows us to entangled factorized states. As an example, if the control qbit is in a superposition state $\vert \psi \rangle = \alpha \vert 0 \rangle + \beta \vert 1 \rangle$  ($\alpha, \beta \not = 0$) and the target is $\vert 0 \rangle$, then $\CNOT$ generates the entangled state 
$$(\alpha \vert 0 \rangle + \beta \vert 1 \rangle) \otimes \vert 0 \rangle \mapsto \alpha (\vert 0 \rangle \otimes \vert 0 \rangle) + \beta  (\vert 1 \rangle \otimes \vert 1 \rangle). $$

In this section we shall study a generalization of this situation. More precisely, we characterize the input $\rho \otimes \sigma$ for which $\QCNOT$ generates a non-factorizable state.

\begin{theo}
Let $\rho$ and $\sigma$ be two density operators in ${\mathbb C}^2$. Then $\QCNOT(\rho\otimes \sigma)$ is factorizable iff  one of the following two conditions holds: 

\begin{enumerate}

\item
$\rho = 
  \left(\begin{array}{cc}
    a_1 & 0 \\
    0 & 1-a_1
  \end{array}\right)$ and $\sigma =
  \left(\begin{array}{cc}
    \frac{1}{2} & b \\
    b & \frac{1}{2}
  \end{array}\right)$,

\item
$\rho = P_i $ where $i \in \{0,1\}$, 

\item
$\sigma = \frac{1}{2}
  \left(\begin{array}{cc}
    1 & \pm 1 \\
    \pm 1 & 1
  \end{array}\right)$.

\end{enumerate}

\end{theo}

\begin{proof}
Consider the following two generic density operators in ${\mathbb C}^2$  
$$
\rho =
  \left(\begin{array}{cc}
    a_1 & a \\
    a^* & 1- a_1
  \end{array}\right) \hspace{0.2cm} and
  \quad
  \sigma =
  \left(\begin{array}{cc}
    b_1 & b \\
    b^* & 1- b_1
  \end{array}\right).
$$
By Proposition \ref{decomposition}, $\QCNOT(\rho\otimes \sigma)$ has the following form: $$\QCNOT(\rho\otimes \sigma) = \QCNOT(\rho\otimes \sigma)_1 \otimes \QCNOT(\rho\otimes \sigma)_2 + {\bf M}(\QCNOT(\rho\otimes \sigma)).$$ Thus we have to establish conditions on  $\rho$, $\sigma$ such that ${\bf M}(\QCNOT(\rho\otimes \sigma)) = 0$. Since ${\bf M}(\QCNOT(\rho\otimes \sigma)) = \QCNOT(\rho\otimes \sigma) - \QCNOT(\rho\otimes \sigma)_1 \otimes \QCNOT(\rho\otimes \sigma)_2$, is straightforward to see that $${\bf M}(\QCNOT(\rho\otimes \sigma)) = (x_{i,j})_{1\leq i,j \leq 4}$$ where 

\begin{enumerate}
\item[11)] $x_{11} = a_1(1-a_1)(1-2b_1) = -x_{22} = -x_{33} = x_{44}$,

\item[12)] $x_{12} = -2i a_1(a_1 - 1)Im(b)$,

\item[13)] $x_{13} = -a(b^* + 2Re(b)(a_1(2b_1 - 1) - b_1)  )$,

\item[14)] $x_{14} = a(b_1 - 2Re(b)(b^* + 2ia_1 Im(b)))$.

\end{enumerate}

\vspace{0.2cm}

\begin{enumerate}

\item[23)] $x_{23} = -a (b_1 - 1 + 2Re(b)(b - 2ia_1Im(b))) $,

\item[24)] $x_{24} = a(b^* - 2Re(b)(a_1 + b_1 -2a_1b_1) )$.

\end{enumerate}

\vspace{0.2cm}

\begin{enumerate}

\item[34)] $x_{34} = 2ia_1Im(b)(a_1 - 1) $.

\end{enumerate}

\vspace{0.2cm}

The other entries of ${\bf M}(\QCNOT(\rho\otimes \sigma))$ are obtained by the conjugation of the above entries. Let us consider the system of equations 
\begin{equation}\label{SIST}
(x_{i,j} = 0)_{1\leq i,j \leq 4}.  
\end{equation}

Note that $x_{11} = 0$ iff $b_1 = \frac{1}{2}$, $a_1 = 0$ or $a_1 = 1$. We shall study these cases

\begin{enumerate}
\item[1] {\bf Case $b_1 = \frac{1}{2}$}
\end{enumerate}

By $x_{13} = 0$ we have that $-a(b^* - Re(b)) = 0$. Thus we have to consider two subcases, $a = 0$ or $b^* = Re(b)$ i.e. $b\in {\mathbb R}.$

\begin{itemize}
\item[1.1] Note that the conditions $b_1 = \frac{1}{2}$, $a= Im(b) = 0$ is a solution of the system (\ref{SIST}) that characterize the input    
$$
\rho =
  \left(\begin{array}{cc}
    a_1 & 0 \\
    0 & 1- a_1
  \end{array}\right) \hspace{0.2cm} and
  \quad
  \sigma =
  \left(\begin{array}{cc}
    \frac{1}{2} & b \\
    b & \frac{1}{2}
  \end{array}\right).
$$
In this way $\QCNOT(\rho\otimes \sigma) = \left(\begin{array}{cccc}
      \frac{a_1}{2} & a_1b & 0 & 0  \\
      a_1b & \frac{a_1}{2} & 0 & 0  \\
      0 & 0 & \frac{1-a_1}{2} & (1-a_1)b  \\
      0& 0 & (1-a_1)b & \frac{1-a_1}{2}  \\
          \end{array}\right)$ which is factorizable as $\rho\otimes\sigma.$

\item[1.2] $b\in {\mathbb R}$. By $x_{23} = 0$ we have that $-a(-\frac{1}{2} + b^2) = 0$, giving the following three possible cases:

\begin{itemize}
\item $a= 0$ which is the case 1.1.

\item $b= \pm \frac{1}{2}$. It provides solutions to the system (\ref{SIST}) that respectively characterizes the input
$$
\sigma = \frac{1}{2}
  \left(\begin{array}{cc}
    1 & \pm 1 \\
    \pm 1 & 1
  \end{array}\right) \mbox{for an arbitrary $\rho$}.
$$
In this way $\QCNOT(\rho\otimes \sigma) = \frac{1}{2}\left(\begin{array}{cccc}
      a_1 & \pm a_1 & \pm a & a  \\
      \pm a_1 & a_1 & a & \pm a  \\
      \pm a^* & a^* & 1-a_1 & \pm (1- a_1)  \\
      a^*& \pm a^* & \pm (1- a_1) & 1-a_1  \\
          \end{array}\right)$ which is factorizable as 

$$
\frac{1}{2}
  \left(\begin{array}{cc}
    a_1 & \pm a \\
    \pm a^* & 1-a_1
  \end{array}\right) \otimes 
  \left(\begin{array}{cc}
    1 & \pm 1 \\
    \pm 1 & 1
  \end{array}\right).
$$

\end{itemize}

\end{itemize}

\begin{enumerate}
\item[2] {\bf Case $a_1 = 0$}
\end{enumerate}

\noindent
By $x_{13} = 0$ $-a(b^* - Re(b)b_1) = 0$. Thus we have to consider two subcases, $a = 0$ or $b^* = 2Re(b)b_1$

\begin{itemize}
\item[2.1] Note that $a= a_1 = 0$ is a solution of the system \ref{SIST} that characterizes the input: $\rho = P_1$ and arbitrary $\sigma$. 

 In this case $\QCNOT(P_1\otimes \sigma) =P_1\otimes(\sigma_1\sigma\sigma_1)$, where $\sigma_1$ is the Pauli matrix introduced above.

\item[2.2] $b^* = 2Re(b)b_1$. Since $b_1 \in {\mathbb R}$, $b \in {\mathbb R}$ and $b(1-2b_1)= 0$. It provides two possibles subcases $b=0$ or $b_1 = \frac{1}{2}$. 

\begin{itemize}
\item $b=0$. By $x_{14}=0$ we have that $a=0$ or $b_1 =0$. The case $a=0$ is an instance of the case 2.1. If $b_1 = 0$ then, the equation $x_{23} = 0$ forces $a= 0$ which is also an instance of the case 2.1.    

\item $b_1 = \frac{1}{2}$. It is an instance of the case 1.

\end{itemize}
\end{itemize}

\begin{enumerate}
\item[3] {\bf Case $a_1 = 1$}
\end{enumerate}

By $x_{13} = 0$ we have that $-a(b^* + 2Re(b)(b_1 -1 )) = 0$. It provides two possibles subcases: $a= 0$ or $b(2b_1 - 1)=0$ where $b\in {\mathbb R}$.

\begin{itemize}
\item[3.1] $a= 0$, $a_1 = 1$ is a solution of the system \ref{SIST} that characterize the input: $\rho = P_0$ and arbitrary $\sigma$. 

 In this way $\QCNOT(P_0\otimes \sigma) = P_0\otimes \sigma.$

\item[3.2] $b(2b_1 - 1)$ gives two possibilities: $b=\frac{1}{2}$ or $b=0$

\begin{itemize}
\item $b=\frac{1}{2}$ is an instance of the case 1.1.

\item $b=0$. By $x_{14}=0$ we have that $ab_1 = 0$. The case $a=0$ is an instance of the case 3.1. If $b_1 = 0$, by $x_{23} = 0$ follows that $a=0$ which is also an instance of the case 3.1. 

\end{itemize}

\end{itemize}

Thus we have analyzed all possible solution of the system \ref{SIST} characterizing the preservation of factorizability for $\QCNOT$.

\qed
\end{proof}

\end{document}